\documentclass{revtex4-2}
\usepackage{amssymb}
\usepackage{amsmath}
\usepackage{qcircuit}
\usepackage{xcolor}
\usepackage{enumitem}
\usepackage{amsthm}
\usepackage{scrextend}
\usepackage{setspace}
\usepackage{caption}
\usepackage {hyperref}
\hypersetup{colorlinks=true,
            breaklinks=true,
            linkcolor=[rgb]{0, 0.2, 0.5},
            citecolor=[rgb]{0, 0.35, 0},      
            urlcolor=[rgb]{0, 0.2, 0.5},
            pdfpagemode=UseOutlines}
\usepackage{cleveref}
\newenvironment{inequality}
  {\crefalias{equation}{inequality}\begin{equation}}
  {\end{equation}\ignorespacesafterend}
\crefname{inequality}{ineq.}{ineqs.}
\counterwithin{figure}{section}
\newtheorem{theorem}{Theorem}[section]
\newtheorem{lemma}[theorem]{Lemma}

\newtheorem{defn}[theorem]{Definition}
\newtheorem{prop}[theorem]{Proposition}
\newtheorem{cor}[theorem]{Corollary}

\newtheorem{Claim}[theorem]{Claim}
\newenvironment{claimproof}[1]{\par\noindent\underline{Proof:}\space#1}{\hspace{1mm}$\blacksquare$}
\newcommand{\mylabel}[2]{#2\def\@currentlabel{#2}\label{#1}}

\begin{document}

\preprint{APS/123-QED}

\title{Noise-Robustness for Delegated Quantum Computation in the Circuit Model}

\author{Anne Broadbent}%
 \email{Anne.Broadbent@uottawa.ca}
 \author{Joshua Nevin}
 \email{jnevin@uottawa.ca}
\affiliation{University of Ottawa\\%
Department of Mathematics and Statistics\\ %
Ottawa, ON\\ %
Canada, K1N 6N5}

\begin{abstract} 
Cloud-based quantum computing, coupled with the rapid progress in quantum algorithms, brings to the forefront the question of verifiability in delegated quantum computations. In the current landscape of noisy quantum devices, this question must be addressed alongside noise tolerance. In this work, we revisit the circuit-based framework for verifiable quantum computation introduced by Broadbent [Theory of Computing, 2018], and extend it to the setting of server-side noise. Our contribution is an improved upper bound on the noise-tolerance threshold, achieved through a protocol that interleaves computation and test rounds in an indistinguishable manner. This structure enables a concise security proof against arbitrary deviations by the server, while ensuring robustness to realistic noise.

\end{abstract}

\maketitle 
\tableofcontents
\section{Introduction} 

The problem of simulating the behavior of many-particle quantum systems is widely believed to be classically intractable. This  intractability drives much of our interest in quantum computation, as a quantum computational device, if built, would be able to simulate the behavior of such a system and thus outperform classical devices. But this immediately raises a question: Once we build a quantum computer, and in the regime of high classical computational complexity, how can we be sure that such quantum device works as claimed? This question is related  to the context of  delegation of quantum computations: as  large-scale quantum computers are currently extremely expensive and therefore accessible almost exclusively remotely, how can a client be sure that the quantum server is actually performing the desired computation? To phrase this differently, how can a provider of cloud quantum services convince potential clients that their quantum computer does what it claims it does?  To add to the complexity of this question, we note that, in the current era of Noisy Intermediate-Scale Quantum (NISQ) computers (see \cite{Pre18}), noise is inherently unavoidable, meaning that any method to verify quantum computations would need to be able to deal with a certain level of noise, even for honest parties. This scenario has been studied in various works (see \cite{KD19,LMKO21,BWM+25arxiv, ArxPrLeaksFaultTol}). The contribution of this paper is to show that the protocol of \cite{Bro18}, which originally assumed perfect hardware, can be made resistant to noise with an improved noise-tolerance threshold over these previous works. This has the potential to enable new experimental demonstrations in today's regime of noisy delegated quantum computation. We also provide a novel proof of soundness of our noise-resistant protocol, showing that there is a relatively simple way to lift the soundness proof of \cite{Bro18} to the noise-robust setting. 

\section{Contributions and Techniques}\label{ComTechSec}

We now give a very high-level overview of the main result and outline the organization of the remainder of the paper. We recall a few fundamental definitions before proceeding. A family $\mathcal{C}_Q=\{C_n: n\in\mathbb{N}\}$ of quantum circuits is called \emph{polynomial-time uniform} if there exists a polynomial-time classical Turing machine, which, on input $n$, outputs a classical description of the $n$th circuit in the family. The class \textsf{BQP} consists of the decision problems that can be solved with a bounded probability of error using a polynomial-time uniform family of quantum circuits. More precisely:

\begin{defn} \emph{\textsf{BQP} consists of the languages $\mathcal{L}\subseteq\{0,1\}^*$ such that there is a uniform polynomial-time family $\{C_n: n\in\mathbb{N}\}$ of quantum circuits, where any $n\in\mathbb{N}$, all three of the following hold.}

\begin{enumerate}[label=\emph{\arabic*)}]
\item\emph{$C_n$ takes $n$ qubits as input and outputs a single bit.}
\item\emph{For any $x\in L$ with $|x|=n$, we have $\textnormal{Pr}(C_n(x))=1)\geq 1-q$.}
\item\emph{For any $x\not\in L$ with $|x|=n$, $\textnormal{Pr}(C_n(x)=1)\leq q$.}
\end{enumerate}

\end{defn}

The parameter $q$ in the definition above, usually taken as $q=\frac{1}{3}$, is arbitrary as long as it is fixed and $q<\frac{1}{2}$ (in particular, since, for the purposes of this paper, we are regarding $q$ as a constant, we automatically get that $q$ is bounded away from $\frac{1}{2}$). We call $q$ the \emph{inherent error probability} of the BQP computation.

Since we restrict ourselves to the setting in which we demand verification, for the remainder of this paper, in keeping with the usual terminology for interactive proofs  (see \Cref{ClassProofSystemSub}), we refer to the client and server as the \emph{verifier} and \emph{prover} respectively. In the prover-verifier setting, we want to design a protocol such that the verifier, with high probability, can be satisfied that the prover's output is indeed the output of the computation they initially had in mind, even though the prover is an untrusted party. This is referred to as the problem of \emph{delegated quantum computation}. In order to model such noise, following a common approach which is to interleave computation rounds with test rounds, which have pre-determined outcomes from the point of view of the verifier. We let $p^{\textnormal{noise}}$ be the probability of an honest prover obtaining an incorrect outcome in a test round. While it is not possible to construct a protocol which both produces correct answers with high probability and tolerates $p^{\textnormal{noise}}$ being arbitrarily high, we address the question of the tradeoff between $p^{\textnormal{noise}}$, the number of iterations, and the quality of the conclusion that can be drawn. Our main result, roughly speaking, is that there is an $N$-round interactive proof protocol between an almost-classical verifier $\mathbf{V}$ and a possibly-noisy polynomial-time quantum computer $\mathbf{P}$ with inherent error probability $q$ such that the $p^{\textnormal{noise}}$ that can be tolerated is $\left(1-O\left(\frac{1}{\sqrt{N}}\right)\right)\frac{2q-1}{2q-2}$, i.e. we can get arbitrarily close to $\frac{2q-1}{2q-2}$ from below, at the cost of more rounds. The remainder of this work is structured as follows:

\begin{enumerate}[label=\arabic*)]
    \item In \Cref{QIntMBQC}, we review previous literature on verified delegated quantum computation, both with perfect hardware and with noise. 
    \item In \Cref{PerfHardwareSec}, we review the protocol of \cite{Bro18} in detail, since this is the perfect-hardware protocol we will be building upon.
    \item In \Cref{MultiRoundSec}, we formally define our noise-resistant modification of the protocol of \cite{Bro18} and show that our problem can be cast as a simple combinatorial game, and, in \Cref{AttackDeviationSec}, we make explicit the connection between the error model of this game and the error model of \cite{Bro18}.
    \item In \Cref{ProbLemmaSec}, we review some basic probability notions which we need in order to analyze this combinatorial game defined in \Cref{MultiRoundSec} and, in particular, to prove our main result.
    \item Finally, in \Cref{NoiseResProtocol}, we prove our main result, which is stated formally as \Cref{MainThmMultiRunStatement}.
    \item The proof of \Cref{MainThmMultiRunStatement} requires a self-contained probability lemma, \Cref{IfWSmallLowerBound}, which abstractly captures our error model. To avoid interrupting the flow of the main proof, the proof is this lemma 
    is placed in the appendix. 
\end{enumerate}

\section{Related Work}\label{QIntMBQC}

The  problem of verifying quantum computations with noisy hardware has been studied in previous work (\cite{KD19, LMKO21} and, very recently, \cite{BWM+25arxiv}). Additionally, \cite{FK17} and \cite{BFK10} also provide protocols for delegated computation with perfect devices. Both the measurement-based protocol in \cite{FK17} and the circuit-based protocol in \cite{Bro18} are extensions of protocols that provide privacy in their respective models \cite{BFK09,Bro15}. Early experimental demonstrations have demonstrated privacy \cite{BKB+12,FBS+14} and verifiability (\cite{BFKW13}, \cite{PublishedVBlindTrap},  \cite{WSS+25}). In particular, \cite{PublishedVBlindTrap} is a real-word implementation of the protocol of \cite{LMKO21} which consists of a trapped-ion quantum server and a photonic client that is able to steer a network qubit just by measuring the polarization of their photon in one of a fixed number of bases. In this way, the (almost) classical client can implement a verifiable blind quantum computation on a brickwork state possessed by the server. 

The protocol of \cite{LMKO21} is of particular importance for us for three reasons: Firstly, our task is very closely related to that of \cite{LMKO21}. Secondly, we use several key ideas from \cite{LMKO21} in our adaptation of the protocol of \cite{Bro18} and, finally, we explicitly improve on the performance of the protocol of \cite{LMKO21}. We now  provide an overview of the protocol of \cite{LMKO21} and some brief analysis of its performance for comparison purposes. The work of \cite{LMKO21} studies the problem of delegating quantum computation in the \emph{measurement-based quantum computation} (MBQC) model. The precise details of the MBQC model are not essential for our purposes, but can be found in \cite{DK06} and \cite{Joz06arxiv}. Briefly, given a finite graph $G=(V,E)$ on $N$ vertices, we construct an $N$-qubit quantum state $|G\rangle$ consisting of a qubit in the $|+\rangle$ state for each vertex, where we apply a $\textsf{CZ}$-gate between two qubits if and only if there is an edge between the corresponding vertices. The input to a computation then consists of a graph $G$, a flow $f$ on $G$ (see \cite{DK06}), and a set of $\{\phi_v: v\in V\}$ of angles initially specifying the bases in which each qubit is to be measured. A quantum computation can then be cast as an adaptive sequence of measurements on the qubits of $|G\rangle$, where the order in which the measurements are performed is specified by $f$ and the measurement bases that are updated as the computation progresses through the flow. 

\bigskip 

The input to the prover-verifier protocol of \cite{LMKO21} is a graph $G=(V,E)$, a flow $f$ and a family $\{\phi_v: v\in V\}$ of initial measurement bases as above. Together, the prover and verifier execute an $n$-round protocol, where each round consists of an MBQC on $|G\rangle$. The underlying graph $G$ is an input to the protocol which is fixed across all rounds. Furthermore, the verifier initially chooses a random partition $(C,T)$ of $[n]$ (that is, $C\cup T=[n]$ as a disjoint union). They also choose a $k$-coloring of $G$. Recall that a subset $S$ of $V$ is called an \emph{independent set} if there is no edge of $G$ with both endpoints in $S$, and a \emph{$k$-coloring of $G$} is a partition of the vertex set $V$ into $k$ pairwise-disjoint subsets $V_1, \cdots, V_k$, where $V=V_1\cup\cdots\cup V_k$, such that, for each $j\in [k]$, $V_j$ is an independent subset of $V$.

\bigskip

The verifier keeps the sets $C,T$ secret, as well as the $k$-coloring. Letting $d,t$ denote $|C|$ and $|T|$ respectively, the parameters $d,t,k$ are published and thus known to the prover. The sequence of rounds is randomly partitioned into a set of \emph{computation rounds} and a set of \emph{test rounds}, but the prover does not know which is which and is provably unable to distinguish between them. For each test round, the verifier secretly chooses a color $j\in [k]$, corresponding to an independent subset of $V$. For this test round, the vertices of $V$ with color class $j$ are \emph{trap qubits} on which measurements have predetermined outcomes known to the verifier. 

\bigskip 

After all $n$ rounds of the protocol are completed, the verifier has a string of classical outcomes $\{y_j: j\in C\}$, where $y_j$ is a classical bit string corresponding to the sequence of final measurement outcomes on the output qubits in round $j$. The verifier first counts the number $w$ of test rounds which were not passed, and, if the fraction $w/t$ is above a certain threshold, they abort the protocol. If $w/t$ is below the threshold, then the verifier looks at $\{y_j: y\in C\}$, and, if there exists a bit string $y$ such that $\#\{j\in C: y_j=y\}\geq d/2$. Otherwise, they again abort the protocol. The test rounds interleaved with the computation rounds are only there for the purpose of catching a potentially cheating prover, but even if we can guarantee that the prover is honest, we still need majority voting on the computation rounds to guarantee that there is a low probability that non-malicious noise has caused the verifier to accept a corrupted outcome from the computation rounds. Informally, the purpose of this paradigm of  is to ensure that, if the level of noise is below a certain (constant) threshold, then \emph{both} of the following hold:

\begin{enumerate}[label=\roman*)]
\item There is a low probability that a majority computational output is not found (i.e. a low probability of aborting); AND
\item If a majority computational output $y$ is found, then there is a low probability that this output is erroneous.
\end{enumerate}
\bigskip
This requires us to make the idea of \emph{noise} somewhat more precise. In \cite{LMKO21}, noise is modelled as a sequence of round-dependent arbitrary CPTP maps, i.e.~a possibly different arbitrary CPTP map acting on each round, with a probability parameter $p^{\textnormal{noise}}$, an upper bound on the probability that at least one of the trap measurements fails in any single test round. Informally the main result of \cite{LMKO21} states that, letting $n=d+t$, both of i) and ii) above hold as long as both $w/t\in\left(0, \frac{1}{k}\cdot\frac{2q-1}{2q-2}\right)$ and $p^{\textnormal{noise}}<\frac{1}{k}\cdot\frac{2q-1}{2q-2}$. As noted in \cite{LMKO21}, for a fixed graph state, using fewer colors is desirable, since it both allows a larger value of $p^{\textnormal{noise}}$ and requires fewer rounds to obtain specified security and correctness targets. Of course, for a nontrivial quantum computation on a graph state, the best possible chromatic bound is $k=2$. This is sufficient for universal quantum computation as well, as there are universal families of graph states whose corresponding graphs are all 2-colorable (see \cite{BFK10}). However, if we have a given graph state in mind, which is not necessarily 2-colorable, converting this to a 2-colorable graph state that executes the same computation potentially requires a large scaling in the number of qubits. On the other hand, our main result ---  \Cref{MainThmMultiRunStatement} --- behaves statistically, in a way that corresponds to a ``$k=1$"-version of the protocol of \cite{LMKO21}, for a computation on \emph{any} graph state, including those that are very far from 2-colorable, without the need to scale to a larger 2-colorable graph state implementing the same computation.

\section{The Perfect-Hardware Protocol of \cite{Bro18}}\label{PerfHardwareSec}

\subsection{Classical Proof Systems}\label{ClassProofSystemSub}

The protocol of \cite{Bro18} is based on \emph{interactive proof systems}. The concept of interactive proof systems was  first introduced in the 1980s in \cite{Bab85} and \cite{GMR89}. Informally, given a language $\mathcal{L}\subseteq\{0,1\}^*$, a proof system is a protocol between a computationally unbounded but not necessarily honest \emph{prover} $\mathbf{P}$ and a computationally bounded but always honest \emph{verifier} $\mathbf{V}$, where $\mathbf{V}$ is usually taken to be a polynomial-time algorithm. For any $x\in\{0,1\}^*$, $\mathbf{V}$ is trying to decide whether to accept $x$ (i.e.~conclude that $x\in\mathcal{L}$ or \emph{reject} $x$ (i.e.~conclude that $x\not\in\mathcal{L}$).

\begin{defn}\label{CompSoundDefn} \emph{Given a language $\mathcal{L}$, a \emph{proof} system for $\mathcal{L}$ is a protocol satisfying the following:}

\begin{itemize}
\item \emph{\textbf{Completeness:} For any $x\in\mathcal{L}$ and any prover $\mathbf{P}$ which honestly follows the protocol, $\mathbf{P}$ convinces $\mathbf{V}$ to accept $x$ with probability $c$.}
\item\emph{\textbf{Soundness:} If $x\not\in\mathcal{L}$, then, for any prover $\mathbf{P}^*$ (where $\mathbf{P}^*$ may arbitrarily deviate from the protocol), the probability that $\mathbf{P}^*$ convinces $\mathbf{V}$ to accept $x$ is at most~$s$.}
\end{itemize}

\end{defn}

\bigskip

The parameters $c,s\in [0, 1)$ are arbitrary. They can always be amplified to be arbitrarily close to one (resp.~zero) by a polynomial amount of repetition as long as $c-s>\frac{1}{\textnormal{poly(n)}}>0$ is satisfied. For the purposes of this paper, we are only ever concerned with the case where $c,s$ are fixed constants, independent of $n$, in which case $c>s$ is sufficient. We refer to $c,s$ as the \emph{completeness (resp.~soundness) parameter} of the interactive proof protocol. If $c$ and $s$ are one and zero respectively, then we refer to this stronger condition as \emph{perfect completeness} (resp.~\emph{perfect soundness}). An \emph{interactive proof system} is a protocol satisfying the above completeness and soundness conditions, in which $\mathbf{P}$ and $\mathbf{V}$ are allowed to exchange messages in multiple rounds until $\mathbf{V}$ either accepts or rejects the initial input string~$x$, where $\mathbf{P}$ is trying to convince $\mathbf{V}$ that $x\in\mathcal{L}$, and $\mathbf{V}$ is trying to determine whether $\mathbf{P}$ has produced a valid proof that $x\in\mathcal{L}$. 

\subsection{Quantum Interactive Proofs in the Circuit Model}\label{QCircIntSub}

The class \textsf{QIP}, the \emph{quantum} analogue of the complexity class \textsf{IP} was introduced by Watrous in \cite{Wat03}. That is, \textsf{QIP} is the set of languages $\mathcal{L}$ for which there exists an interactive proof system in which $\mathbf{V}$ is a polynomial-time quantum verifier and $\mathbf{P}$ is a computationally unbounded quantum prover, but the situation we want to study is somewhat more restricted. A natural way to restrict the above definition is to consider the class of languages $\mathcal{L}$ which admit an interactive proof system in which $\mathbf{P}$ is a  polynomial-time quantum prover and $\mathbf{V}$ is a classical (probabilistic) polynomial-time machine, but, at the time of writing, the problem of verifying an arbitrary quantum computation for a \emph{purely} classical verifier and a \textsf{BQP} prover is still an open problem (however, under computational assumptions, this is possible~\cite{Mah18a}). Instead, we consider the setting in which $\mathbf{P}$ is a \textsf{BQP} prover and $\mathbf{V}$ is a ``hybrid" classical-quantum machine, that is, a classical probabilistic polynomial-time machine augmented with some rudimentary quantum operations. The class of languages defined by interactive proof systems with such a prover-verifier model was first studied in \cite{ABEM17arxivb}, which defined the complexity class \textsf{QPIP}$_{\kappa}$, where $\kappa$ is the number of qubit registers permitted for the hybrid classical-quantum verifier. In this setting, the verifier has a classical part, consisting of a classical probabilistic polynomial-time machine, and a quantum part, consisting of a register of $\kappa$ qubits. The verifier can perform arbitrary quantum computations on the qubits of these $\kappa$ registers. The classical part of the machine controls which quantum operations are to be performed and the measurement outcomes from the quantum part can be accepted as input to the classical part as well.

\bigskip

The class \textsf{QPIP}$_{\kappa}$ can be viewed as capturing a set of languages defined by protocols whose verifier represents ``realistic" computational abilities for devices that we can build now, and the main result of \cite{ABEM17arxivb} was to show that there exists a constant $\kappa$ such that \textsf{BQP}=\textsf{QPIP}$_{\kappa}$.

\bigskip

In \cite{Bro18}, Broadbent studied a complexity class which is a  modification of \textsf{QPIP}$_{\kappa}$. To formally state the result of \cite{Bro18}, we need the following definitions. The analogous results and Definitions of \cite{Bro18} are originally stated in terms of a fixed inherent error probability $q=1/3$ of a \textsf{BQP} computation. For our purposes, we need a slightly generalized version of this terminology so that we can measure how the noise tolerance of our modified protocol scales with $q$. We can then formally stated the main result of \cite{Bro18} in this modified language.

\begin{defn} \emph{A \emph{promise problem} is a pair $(\Pi_Y, \Pi_N)$, where $\Pi_Y$ and $\Pi_N$ are disjoint subsets of $\{0,1\}^*$, corresponding respectively to YES and NO instances of the problem.}  \end{defn}

The particular promise problem we are concerned with is the following.

\begin{defn}\label{QCirProProbDef} \emph{We define the promise-problem Q-CIRCUIT$_q$ as follows: Let $U$ be a quantum circuit consisting of a sequence of gates drawn from the universal gate set $\{\textsf{X, Z, H, CNOT, T}\}$. We set}
$$\mathcal{P}(U):=|\langle\mathbb{I}_{n-1}\otimes |0\rangle\langle 0|, U|0\rangle^{\otimes n}\rangle|^2$$
\emph{This is the probability of observing $|0\rangle$ in the $n$th register after making a measurement of the $n$th qubit of $U|0\rangle^{\otimes n}$ in the computational basis. Then:}
\begin{enumerate}[label=\emph{\arabic*)}]
    \item\emph{$Q\textnormal{-CIRCUIT}_{Y, q}$ is the set of $n$-qubit quantum circuits $U$ with $\mathcal{P}(U)\geq 1-q$}
    \item\emph{$Q\textnormal{-CIRCUIT}_{N, q}$ is the set of $n$-qubit quantum circuits $U$ with $\mathcal{P}(U)\leq q$}
\end{enumerate}

\end{defn}

The problem above is \textsf{BQP}-complete (see \cite{ABEM17arxivb}) which makes it a natural setting to capture arbitrary \textsf{BQP} computations. 

\begin{defn}  \emph{Let $\mathcal{S}=\{\mathcal{S}_1, \cdots, \mathcal{S}_{\ell}\}$, where each $\mathcal{S}_1, \cdots, \mathcal{S}_{\ell}$ is a set of density operators. We say that a classical polytime machine $\mathbf{V}$ is $\mathcal{S}$-\emph{augmented} if it is also given the ability to randomly generate states from each of $\mathcal{S}_1, \cdots, \mathcal{S}_{\ell}$. That is, $\mathbf{V}$ is augmented with the ability to pick an index $i$ in $\{1, \cdots, \ell\}$ and then randomly generate one of the elements of $\mathcal{S}_i$.} \end{defn}

\begin{defn}\label{SQuantIntProvSysDef} \emph{Let $\mathcal{S}=\{\mathcal{S}_1, \cdots, \mathcal{S}_{\ell}\}$, where each $\mathcal{S}_1, \cdots, \mathcal{S}_{\ell}$ is a set of density operators. Let $\Pi=(\Pi_Y, \Pi_N)$ be a promise problem. Let $c,s\in [0,1)$ be parameters with $c>s$. An \emph{$(\mathcal{S}, c, s)$-quantum-prover interactive proof system} for $\Pi$ consists of a prover $\mathbf{P}$ and a verifier $\mathbf{V}$ such that}
\begin{enumerate}[label=\emph{\arabic*)}]
\item \emph{$\mathbf{V}$ is an $\mathcal{S}$-augmented classical polytime machine. Furthermore, when $\mathbf{V}$ picks an  index $i$ in $\{1, \cdots, \ell\}$ and then randomly generates one of the elements of $\mathcal{S}_i$, the index $i$ is then sent to $\mathbf{P}$, but the element of $\mathcal{S}_i$ chosen by $\mathbf{V}$ is kept secret from $\mathbf{P}$}
\item\emph{$\mathbf{P}$ is a quantum polytime machine}
\end{enumerate}
\emph{and, furthermore, the following hold}
\begin{enumerate}[label=\emph{\alph*)}]
\item\emph{If $x\in\Pi_Y$, then $\mathcal{P}$, behaving honestly, convinces $\mathbf{V}$ of this fact with probability $\geq c$}
\item\emph{If $x\in\Pi_N$, then there is no prover $\mathbf{P}^*$, even a computationally unbounded one, which can convince $\mathbf{V}$ that $x\in\Pi_Y$ with probability $\geq s$}
\end{enumerate}
\emph{Given $\mathcal{S}, c, s$ as above, we define QPIP$_{\mathcal{S}, c, s}$ to be the set of promise problems $\Pi$ such that there is an $\mathcal{S}$-quantum-prover interactive proof system for $\Pi$.}
\end{defn} 

In particular, in the complexity class in \Cref{SQuantIntProvSysDef}, we are still considering a hybrid quantum-classical verifier as in \cite{ABEM17arxivb}, but the set of quantum operations the verifier is permitted has been restricted even further. In the setting of \Cref{SQuantIntProvSysDef}, the verifier no longer has any ability to perform quantum computations. Instead, the verifier can only prepare and send qubits drawn from some pre-specified finite set, but they have no quantum memory or quantum processing power. The main result of \cite{Bro18} is as follows.

\begin{theorem}\label{MainThDelegQComP} Let $$\mathcal{S}:=\{\{|0\rangle, |1\rangle\}, \{|+\rangle, |-\rangle\}, \{\emph{\textsf{P}}|+\rangle, \emph{\textsf{P}}|-\rangle\}, \{\emph{\textsf{T}}|+\rangle, \emph{\textsf{T}}|-\rangle\}, \{\emph{\textsf{PT}}|+\rangle, \emph{\textsf{PT}}|-\rangle\} \}$$
Then there is an $(\mathcal{S}, \frac{8}{9}, \frac{7}{9})$-quantum interactive proof system for the promise problem \textnormal{Q-CIRCUIT}$_{\frac{1}{3}}$. Thus, for any fixed constants $c,s$ with $c>s$, we have $\emph{\textsf{BQP}}=\emph{\textsf{QPIP}}_{c,s}$.

\end{theorem}

The next section provides a brief overview of the method to achieve \cref{MainThDelegQComP}. This is necessary for our purposes because we will be relying to some extent on the inner workings of \cite{Bro18} in order to show the security of the noise-robust version. 

\subsection{A Brief Overview of the Protocol of \cite{Bro18}}\label{SketchSubSection}

Suppose that we have a classical polytime verifier $\mathbf{V}$ who has a particular unitary $U$ on $n$ qubits, and they wish to decide whether $U$ is a YES-instance or a NO-instance of Q-CIRCUIT$_{1/3}$. The prover and verifier run an interactive proof protocol, denoted IPS1. The input to IPS1 is an $n$-qubit quantum circuit~$\mathcal{C}$ whose 1- and 2-qubit gates are drawn from $\{\textsf{X, Z, CNOT, H, T}\}$. The verifier sends the circuit description to the prover (i.e.~the circuit itself is not kept secret). To begin IPS1, the verifier chooses at random one of three possible runs to interact with the prover. One of these is a computation run, the other two are test runs. The test runs are used to detect a deviating prover. There are two types of test runs: $\textsf{X}$-test and $\textsf{Z}$-test. The prover cannot distinguish between the three runs and the verifier does not inform the prover of which type of run they have chosen. In a computation run, the prover executes the target circuit on a Pauli-encryption of $|0\rangle^{\otimes n}$. In a $\textsf{X}$-test rounds, the prover is just executing identity on a Pauli-encryption of $|0\rangle^{\otimes n}$, and, in a $\textsf{Z}$-test round, the prover is just executing identity on a Pauli-encryption of $|+\rangle^{\otimes n}$. On any of the three types of rounds, the prover and verifier proceed through the circuit gate-by-gate, with the verifier updating the (secret) classical encryption keys, exchanging classical messages with the prover, and sending encrypted auxiliary qubits to the prover for the $\textsf{T}$-gate gadgets. At the end of a single round of IPS1 which is either a computation round or an $\textsf{X}$-test-round, the verifier is holding a single classical output bit $c$. If the round is an $\textsf{X}$-test-round, then the expected predetermined outcome $b\in\{0,1\}$ of the run is known to the verifier, and, if the prover is honest, then $b=c$. For $\textsf{Z}$-test-rounds, the verifier does not  check the output bit, they only check some of the internal measurements, whose expected outcomes are also known to the verifier. For the purposes of our abstract error model, this distinction between $\textsf{X}$- and $\textsf{Z}$-test rounds does not actually matter, because our $p^{\textnormal{noise}}$ is just a measure of the probability that, on interaction with an honest prover, the verifier detects at least one error in a given test round. 

\bigskip

The proof of \emph{completeness} of the protocol against an honest prover in \cite{Bro18} is relatively straightforward (we just verify that the gadgets and key updates produce the correct outcome, as if $\mathbf{V}$ had executed the target circuit themselves). The proof of \emph{soundness} against an arbitrary deviating prover is more challenging. Given that the prover and verifier exchange classical messages, the prover receives encrypted data and auxiliary qubits from the verifier, and the prover both executes unitaries and performs measurements, it seems difficult and cumbersome to model an arbitrary attack from the prover in such a way as to provide a rigorous proof of soundness. The essential proof technique in \cite{Bro18} is to introduce an intermediate proof protocol whose soundness parameter upper bounds the soundness parameter of IPS1, where, in this intermediate proof protocol, $\mathbf{V}$ is augmented with more power and a general attack model is easier to model. Some brief discussion of this is necessary for our purposes. In particular, \cite{Bro18} defines an EPR-based version of IPS1, denoted IPS2, in which:
\begin{enumerate}[label=\arabic*)]
\item All quantum inputs (i.e.~both encrypted auxiliary qubits and data qubits) sent by $\mathbf{V}$ are half-EPR pairs, where one wire is entangled with one of $\mathbf{P}$'s registers and the other is held by $\mathbf{V}$, and $\mathbf{V}$ delays all choices (of which rounds are test and computation rounds) until after all $N$ rounds are complete and there is no further communication with the prover; AND
\item All classical messages sent by the verifier are random bits (i.e.~$\mathbf{V}$ no longer sends encryptions of the measurement outcomes).
\end{enumerate}
\cite{Bro18} also defines another interactive proof protocol, denoted IPS3, which is the same as IPS2, except that $\mathbf{V}$ does all the measurements for the interactive gadgets and $\mathbf{P}$ is a \emph{unitary} prover. IPS1 and IPS2 have the same completeness parameter, and the soundness parameter for IPS3 is an upper bound on the soundness parameter for IPS2, which, in turn, is the same as the soundness parameter for IPS1.

The fact that the soundness parameter of IPS3 is an upper bound in the soundness parameter for IPS1 has another useful cryptographic application: In certain contexts, it is necessary, when performing delegated quantum computation, to guard against leakage of the classical secret keys (see \cite{ArxPrLeaksFaultTol}). In obtaining the security bounds of \Cref{MainThmMultiRunStatement}, we do not need to consider this because the measurements (whose random outcomes one-time-pad the classical encryption keys) are conducted by the verifier after all of the rounds of the protocol are over. 

\subsection{Completeness and Soundness Parameters for \cite{Bro18}}

We note that any BQP algorithm $\mathcal{C}_Q=\{C_n: n\in\mathbb{N}\}$ with inherent error probability $q$ can always be transformed into an instance of Q-CIRCUIT$_q$. To see this, let $\mathcal{L}\subseteq\{0,1\}^*$ be the language corresponding to $\mathcal{C}_Q$. Given an $x\in\{0,1\}^n$, letting $n:=|x|$, deciding whether $x\in L$ with error probability $\leq q$ is equivalent to deciding whether $C_nP_x$ lies in Q-CIRCUIT$_{Y, p}$ or Q-CIRCUIT$_{N,q}$, where $P_x$ is the Pauli (consisting of $X$-gates on some wires) mapping $|0\rangle$ to computational basis state $|x\rangle$. The promise problem in \Cref{QCirProProbDef} is a natural choice for determining how the noise threshold of a multi-round run of the protocol IPS1 of \cite{Bro18} scales with the inherent error probability of the computation, because it is straightforward to slightly modify the proof of \Cref{MainThDelegQComP} to show that IPS1 provides an interactive proof protocol with constant completeness-soundness gap for any fixed inherent error probability~$q$. 

\begin{prop}\label{FixedInErrorProb} Let $\mathcal{S}$ be as in \Cref{MainThDelegQComP} and let $0<q<\frac{1}{2}$ be an arbitrary constant. A single run of IPS1 (with perfect hardware) provides an $(\mathcal{S}, c, s)$-interactive proof system for the promise problem \emph{Q-CIRCUIT}$_q$, where $c=1-\frac{q}{3}$ and $s=\frac{1-q}{3}$. In particular, we have an $\Omega(1)$ completeness-soundness gap, since 
$$c-s=\left(1-\frac{q}{3}\right)-\frac{q+2}{3}=\frac{1-2q}{3}>0$$

\end{prop}

\section{The Multi-Round Run of the protocol of \cite{Bro18} as a Combinatorial Game}\label{MultiRoundSec}

In this section, we first formally define our multi-run version of the protocol of \cite{Bro18} and then the model we use to analyze it.

\subsection{IPS\texorpdfstring{$k$}{k} with abort conditions}

In a single run of IPS1, only two outcomes are possible (either $\mathbf{V}$ accepts or rejects the bit). Analogous to \cite{LMKO21}, for our multi-round run of IPS1, we allow a third possible outcome ($\mathbf{V}$ aborts the round). This is essential now that we allow noise in $\mathbf{P}$'s hardware. We cannot prevent a malicious prover from forcing us to constantly abort, but we can simultaneously bound the probability that an honest prover forces us to abort, and the probability that a potentially malicious prover convinces us of the wrong outcome without forcing us to abort. Note that our modified version of the protocol of \cite{Bro18} differs from the original in another way: It is symmetric in YES- and NO-instances of the promise problem. 

\bigskip

\begin{defn}\label{ImplemNSRUN} \emph{Let $N\geq 1$ be an integer and let $w\in [0, 1]$. Let $\mathcal{C}$ be a quantum circuit with gates drawn from $\{\textsf{X, Z, H, CNOT, T}\}$ and let $\mathcal{S}$ be as in \Cref{MainThDelegQComP}. We define a \emph{$(N, w, \mathcal{C})$-run} of IPS1 between an $\mathcal{S}$-augmented classical verifier $\mathbf{V}$ and a $\textsf{BQP}$-prover $\mathbf{P}$ with inherent error probability $q$ to consist of running IPS1 $N$ times in the following way: We produce a random partition $(S_C, S_{\textsf{X}}, S_{\textsf{Z}})$ of the set $\{1, \cdots, N\}$, by rolling a fair 3-sided die, into respective sets of computation rounds, $\textsf{X}$-test rounds, and $\textsf{Z}$-test rounds. Then,}
\begin{itemize}
\item \emph{the verifier \textit{accepts} $\mathcal{C}$ as belonging to Q-CIRCUIT$_{Y,q}$ if and only if both of 1)-2) below hold.}
\begin{enumerate}[label=\emph{\arabic*)}]
\item\emph{The verifier obtains the classical outcome $0$ in at least half of the computation rounds}
\item\emph{For each $\textsf{A}\in\{\textsf{X}, \textsf{Z}\}$, strictly less than a $w$-fraction of the $\textsf{A}$-test rounds are failed, where a test round is said to be \emph{failed} if at least one error is detected by the verifier in that round.}
\end{enumerate}
\item \emph{Likewise, the verifier \emph{rejects} $\mathcal{C}$ as belonging to Q-CIRCUIT$_{Y,q}$ if and only if both: 2) above holds but the verifier obtains the classical outcome $0$ in strictly fewer than half of the computation rounds.}
\item \emph{Lastly, if, for some $\textsf{A}\in\{\textsf{X}, \textsf{Z}
\}$, at least a $w$-fraction of the $\textsf{A}$-test-rounds is failed, then the verifier aborts.} 
\end{itemize}
\end{defn}

Note that 1) corresponds to the typical majority voting condition that is used to bound the probability of error of a protocol to within any desired threshold.  We define analogous multi-run versions of the intermediate protocols of \cite{Bro18}.

\begin{defn} \emph{For $k=2,3$, we define a \emph{$(N, w, \mathcal{C})$-run} of IPS$k$ analogously to \Cref{ImplemNSRUN}, running the auxiliary protocol IPS$k$ in place of IPS1.}\end{defn}

\subsection{The $(\alpha, w)$-avoidance game}\label{alphawAvoidSubsec}

In order to upper bound the probability of a possibly malicious prover convincing a verifier of the wrong answer in the setting of the protocol of \Cref{ImplemNSRUN}, we now describe an abstract attack model from a deviating prover against this protocol. We consider a simple combinatorial game played between two players (who we will call $\mathbf{V}$ and $\mathbf{P}$) which works as follows. Given parameters $\alpha, w\in [0, 1/2)$, we define the $(\alpha, w)$-\emph{avoidance game} to consist of the following:

\begin{enumerate}[label=G\arabic*)]
    \item $\mathbf{V}$ and $\mathbf{P}$ first agree on an integer $N\geq 1$
    \item $\mathbf{V}$ generates a subset $C$ of $[N]$ by a sequence of $N$ independent Bernoulli trials, each with success probability $\frac{1}{3}$ (say, by rolling a fair 3-sided die with one side marking the elements of $C$) and keeps this $C$ secret from $\mathbf{P}$.
    \item Now $\mathbf{P}$ chooses a subset $S\subseteq\{1, \cdots, N\}$
    \item $\mathbf{P}$ wins the game if both of the following hold:
    \begin{enumerate}[label=\roman*)]
        \item $|S\cap C|>\alpha |C|$
        \item $|S\setminus C|\leq w(N-|C|)$
    \end{enumerate}
Otherwise, $\mathbf{V}$ wins the game.
\end{enumerate}

In the setting of the $(\alpha, w)$-avoidance game, we always think of $\alpha$ as corresponding to $\frac{1-2q}{2(1-q)}\in (0, 1/2]$ for a \textsf{BQP} computation with inherent error probability $q\in [0, 1/2)$. Thus, our main task, in order to prove \Cref{MainThmMultiRunStatement}, is to show that, for $N$ sufficiently large, there exists a $w$, as a function of $\alpha$, such that the probability of $\mathbf{P}$ winning the $(\alpha, w)$-avoidance game with preselected integer $N$ is exponentially small in $N$. We prove this in the appendix in \Cref{IfWSmallLowerBound}, where we actually consider a slightly more general setting in which $\mathbf{P}$ has some source of randomness allowing them to select some elements of $S\setminus C$ ``without detection" (this corresponds to the parameter $\epsilon$ in \Cref{IfWSmallLowerBound}). 

The $(\alpha, w)$-avoidance game can be viewed as capturing the simplest possible attack model from a deviating prover against the protocol in \Cref{ImplemNSRUN}, or, more generally, against any circuit-based multi-round run of prover-verifier protocol $\mathcal{P}$ that interleaves test rounds with computation rounds. In this setting, the verifier is implementing a sequence of $N$ rounds of $\mathcal{P}$ and the prover is trying to choose a set $S\subseteq\{1, \cdots, N\}$ of rounds to corrupt so that they simultaneously corrupt a high fraction of the computation rounds and are caught deviating on few enough test rounds that the verifier does not abort the protocol. If the inherent error of the \textsf{BQP} computation is $q$, then, to convince the verifier of the wrong outcome, how many computation rounds need to be corrupted? That is, suppose we are in an instance where we are trying to decide whether a circuit $\mathcal{C}$ lies in Q-CIRCUIT$_{Y,q}$ or Q-CIRCUIT$_{N,q}$ via repetition of $\mathcal{P}$. If there are $N$ rounds in total, of which $|C|$ are computation rounds, and $N'$ of these $C$ rounds remain uncorrupted (i.e. $N'=|C\setminus S|$) then we observe the wrong outcome on at most $qN'+|C\cap S|$ rounds, so we require that
$$qN'+|C\cap S|\geq\frac{|C|}{2}$$
Equivalently,
\begin{equation}\label{EquivFrac}(1-q)|C\cap S|\geq\frac{|C|(1-2q)}{2}\end{equation}
Thus, for the malicious prover to succeed, we require that $|S\cap C|>\frac{|C|(1-2q)}{2(1-q)}$. With the definitions of \Cref{MultiRoundSec} in hand, we already have enough to formally state our main result, \Cref{MainThmMultiRunStatement}, but, before proving this result, we need to make explicit how the model of \Cref{MultiRoundSec} relates to the attack model of \cite{Bro18}, where, the in the latter setting, we are considering attacks from the prover on individual computation and test rounds (that is, we are reasoning at the level of Kraus operators).

\section{Modelling Attacks from the Prover}\label{AttackDeviationSec}

\subsection{Strategies against the $(\alpha, w)$-avoidance game for the malicious prover}

In the $(\alpha, w)$-avoidance game, it is clear that, for $N$ sufficiently large, if $\mathbf{P}^*$ is an arbitrary unbounded quantum player, then neither additional classical randomness nor the quantum hardware that $\mathbf{P}^*$ has access to can improve the probability that $\mathbf{P}^*$ wins this game beyond the bound in \Cref{IfWSmallLowerBound}, i.e.~$\mathbf{P}^*$ can do no better than simply picking an element of $2^N$ at random (which is precisely the strategy considered in \Cref{IfWSmallLowerBound}) because, firstly, any classically random strategy for $\mathbf{P}^*$, i.e. ~a probability distribution $\mathcal{P}: 2^N\rightarrow [0, 1]$, is a convex combination of strategies in which the prover picks a single random (and unknown to the verifier) subset of $\{1, \cdots, N\}$. Secondly, the game is ``one-shot", that is, except for an $N$ agreed upon in advance, there is no interaction between the players until $S$ and $C$ are revealed at the end. However, we also need to show that in the setting of \Cref{ImplemNSRUN}, any attack from an arbitrary unbounded deviating prover $\mathbf{P}^*$ has a success probability (of convincing $\mathbf{V}$ to accept the wrong answer and not abort) that is upper bounded by the probability that $\mathbf{P}^*$ wins the $(\alpha, w)$-avoidance game. This is not immediately obvious because the game is a one-shot protocol in which the players exchange no messages, whereas the protocol of \Cref{ImplemNSRUN} involves sequential exchange of $O(N)$ classical and quantum messages through which $\mathbf{P}^*$ can potentially learn about which choices $\mathbf{V}$ has made so far. In fact, however, using the auxiliary protocols discussed in \Cref{SketchSubSection}, it is clear that the security proof in \cite{Bro18} for a \emph{single} computation run can be lifted to our setting to prove that \Cref{IfWSmallLowerBound} does indeed provide an upper bound on the probability that an attack by a deviating $\mathbf{P}^*$ on the protocol of \Cref{ImplemNSRUN} succeeds. This is because we can consider soundness against an $(N, w, \mathcal{C})$-run of IPS3, and, as noted in \Cref{SketchSubSection}, in IPS3, the verifier can delay the choice of test and computation rounds until after all $N$ rounds are complete and there is no further communication with the prover, so even an unbounded quantum prover can do no better than random guessing in attempting to learn which rounds are computation rounds and which rounds are test rounds, and thus does no better than the interaction-less setting of the $(\alpha, w)$-avoidance game. We note that the above discussion also shows that the multi-round run of the protocol of \cite{Bro18} has the added advantage of simplifying the security proof over that of \cite{LMKO21}, because, in the latter setting, we need to consider the added power obtained from an attack which is entangled across rounds.

\subsection{Attacks on single rounds in \cite{Bro18}}\label{SingleRoundAttackSubsec}

In order to obtained the desired lower bound on the noise tolerance in \Cref{MainThmMultiRunStatement}, we need to be able apply \Cref{IfWSmallLowerBound} with $\epsilon=0$. Informally, in the statement of \Cref{IfWSmallLowerBound}, $\epsilon$ corresponds to the probability that, in a single run of IPS3, a general attack $\Phi$ from a deviating prover $\mathbf{P}^*$ satisfies both of the following. 

\begin{enumerate}[label=\arabic*)]
\itemsep-0.1em
    \item If $\Phi$ is executed on a test round, it is not detected by $\mathbf{V}$; AND
    \item If $\Phi$ is executed on a computation round, it results in a different output distribution on the measurement of the $n$th qubit in the computational basis. 
\end{enumerate}

We now consider the action of a general attack map on a single round of IPS2. For an integer $m\geq 1$, we let $\mathbb{P}_m$ denote the usual Pauli group on $m$ qubits, i.e.~$\mathbb{P}_m$ is the tensor product of $m$ copies of the single-qubit Pauli (phaseless) Pauli group $\mathbb{P}_1:=\{\textsf{I, X, Y, Z}\}$. If the input circuit $\mathcal{C}$ to the protocol has $n$ data qubits and $t$ $\textsf{T}$-gates, then a Pauli attack on a single run of IPS2 is an element of $\mathcal{P}_{n+2t}$. 

\begin{defn} \emph{Let $\mathcal{C}, n, t$ be as above and $m:=n+2t$. For a fixed Pauli $P\in\mathbb{P}_m$, we call $P$ \emph{benign} if the measured qubits in the protocol are acted on only by a gate in $\{\textsf{I, Z}\}$. Conversely, a Pauli $P\in\mathcal{P}_m$ is called \emph{non-benign} if at least one measured qubit in the protocol is acted on only by a gate in $\{\textsf{X, Y}\}$.}\end{defn}

The benign Pauli attack maps have the property that, when they are executed on a computation round, they do not affect the outcome distribution when the last data qubit (i.e.~non-auxiliary qubit) is measured in the computational basis. Conversely, any non-benign attack applied to a test-round will be detected by the verifier. These are both proved formally in \cite{Bro18}, that is, we have the following, obtained from Lemmas 7.8 and 7.4 respectively of \cite{Bro18}.  

\bigskip

\begin{prop}\label{Epsilon0Prot} In a single (perfect-hardware) run of IPS2 on a circuit $\mathcal{C}$ with $n$ data qubits and $t$ $\emph{\textsf{T}}$-gates, the following hold:
\begin{enumerate}[label=\arabic*)]
     \item For any benign attack $\Phi$, if $\Phi$ is executed in a computation round, then the outcome distribution of measuring the final data qubit in the computational basis is the same distribution as when the protocol is followed honestly (i.e. the identity attack is executed); AND
    \item for any non-benign attack $\Phi$, if $\Phi$ is executed on a test run, it is detected by the prover with probability~1, i.e. the test round is failed with probability~1.
\end{enumerate} \end{prop}

\section{Basic Probability Notions}\label{ProbLemmaSec}

To prove \Cref{IfWSmallLowerBound} and \Cref{MainThmMultiRunStatement}, we first recall the following preliminaries.

\begin{defn} (Binomial distribution) \emph{Let $\mathcal{X}$ be a random variable taking values in $\{0, \cdots, n\}$. For $p\in [0,1]$, we write $\mathcal{X}\sim\textnormal{Binomial}(n ,p)$ if $\mathcal{X}$ corresponds to a sequence of $n$ independent trials, each of which succeeds with probability $p$, i.e. $\textnormal{Pr}(\mathcal{X}=k)$ is the probability of precisely $k$ successes.} \end{defn}

\begin{prop}\label{HoeffTailBoundsProp} (Hoeffding tail bounds) Let $\mathcal{X}$ be a random variable taking values in $\{0, \cdots n\}$, where $\mathcal{X}\sim\textnormal{Binomial}(n ,p)$. For $k\leq np$ we have
\begin{inequality}\label{HoeffdingBd1}\textnormal{Pr}[\mathcal{X}\leq k]\leq\textnormal{exp}\left(-2\frac{(np-k)^2}{n}\right)\end{inequality}
Likewise, for $k\geq np$ we have
\begin{inequality}\label{HoeffdingBd2}\textnormal{Pr}[\mathcal{X}\geq k]\leq\textnormal{exp}\left(-2\frac{(np-k)^2}{n}\right)\end{inequality}

\end{prop}

Note that, if $F(k; n, p)=\textnormal{Pr}[\mathcal{X}\leq k]$ denotes the cumulative distribution function for $n$ trials with success probability $p$, then $F(n-k; n, 1-p)=\textnormal{Pr}[\mathcal{X}\geq k]$, so, \cref{HoeffdingBd2} actually just follows from \cref{HoeffdingBd1}. 

\begin{defn} (Hypergeometric distribution) \emph{Given an integer $N\geq 1$ and integers $0\leq K\leq N$ and $0\leq n\leq N$, and a random variable $\mathcal{X}$ taking values in $\{0, \cdots, n\}$, we write $\mathcal{X}\sim\textnormal{Hypergeometric}(N, K, n)$ if, for each $k=0, \cdots, n$, $\textnormal{Pr}[\mathcal{X}=k]$ is the probability of obtaining $k$ successes out of $n$ draws without replacement from a bag with $N$ items, of which $K$ are marked, where the marked items are successes.} \end{defn}

Analogous to \Cref{HoeffTailBoundsProp}, we have the following bounds on the upper and lower tail of the hypergeometric distribution.

\begin{prop}\label{HypergeoTailBoundsProp} Let $\mathcal{X}$ be a random variable with $\mathcal{X}\sim\textnormal{Hypergeometric}(N, K, n)$. Letting $p=K/N$ and $0<t<K/N$, we have
\begin{inequality}\label{HypergeoTail}\textnormal{Pr}[\mathcal{X}\leq (p-t)n]\leq e^{-2t^2n}\end{inequality}
\begin{inequality}\label{HypergeoHead}\textnormal{Pr}[\mathcal{X}\geq (p+t)n]\leq e^{-2t^2n}\end{inequality}
 \end{prop}
 
The bounds of \Cref{HypergeoTailBoundsProp} have the following two corollaries. Note that, for $\mathcal{X}\sim\textnormal{Hypergeometric}(N, K, n)$, the distribution is centered on $\frac{nK}{N}$, and, as we get farther from this mean, the distribution tails off, as the following two corollaries make clear.

\begin{cor}\label{TailBdCor1} Let $\mathcal{X}\sim\textnormal{Hypergeometric}(N, K, n)$ and let $0<\lambda<\frac{nK}{N}$. Then
$$\textnormal{Pr}\left[\mathcal{X}\leq\lambda\right]\leq\textnormal{exp}\left(-2n\left(\frac{K}{N}-\frac{\lambda}{n}\right)^2\right)$$
\end{cor}

Likewise

\begin{cor}\label{TailBdCor2}  Let $\mathcal{X}\sim\textnormal{Hypergeometric}(N, K, n)$ and let $\lambda>\frac{nK}{N}$. Then 
$$\textnormal{Pr}\left[\mathcal{X}\geq\lambda\right]\leq\textnormal{exp}\left(-2n\left(\frac{\lambda}{n}-\frac{K}{N}\right)^2\right)$$\end{cor}

\section{Proof of the Main Result: A Noise-Robust Protocol Obtained from IPS1}\label{NoiseResProtocol}

We can now state and prove our main result. Below, in our statement of conditions on $p^{\textnormal{noise}}, w$, we have suppressed the three arguments of the functions $f,g$, as well as the arguments of $f$ in the definition of $g$, for clarity. 

\bigskip

\begin{theorem}\label{MainThmMultiRunStatement} (Noise-Robust Many-Round Implementation of IPS1 with Majority Voting) There exist functions $f:\mathbb{N}\times [0, 1/2)\times [0,1/2)\rightarrow (0,1)$ and $g:\mathbb{N}\times [0, 1/2)\times [0,1/2)\rightarrow (0, 1)$ such that $$f(N, \alpha, \delta)=1-O\left(\frac{\sqrt{\log(2/\delta)}}{\alpha\sqrt{N}}\right)$$ and $$g(N, \alpha, \delta)=\Theta\left(\frac{\sqrt{\log(2/\delta)}}{\alpha\sqrt{N}}\right)$$ and such that the following holds: Let $\mathcal{S}$ be as in \Cref{MainThDelegQComP}. Let $\mathbf{V}$ be an $\mathcal{S}$-augmented classical polytime machine. Let $\mathcal{C}$ be an $n$-qubit circuit with gates drawn from $\{\emph{\textsf{X, Z, H, CNOT, T}}\}$. Let $\mathbf{P}$ be a possibly noisy polynomial-time quantum computer with inherent error probability $q\in [0, 1/2)$ and let $p^{\textnormal{noise}}$ denote the probability that $\mathbf{P}$, behaving honestly, fails a given single $\emph{\textsf{X}}$- or $\emph{\textsf{Z}}$-test round of IPS1. Let $\delta\in [0, 1/2)$ and $\alpha:=\frac{2q-1}{2q-2}$ and suppose we choose an $N$ sufficiently large such that both $f\geq 9/10$ and $p^{\textnormal{noise}}<f\alpha-g$. In that case, there exists a $w\in (0, 1]$, depending on $N, \alpha, \delta$ such that any $\left(N, w, \mathcal{C}\right)$-run of IPS1 satisfies both of 1) and 2) below. In particular, we take $w=f\alpha$.
\begin{enumerate}[label=\arabic*)]
\item There is a probability of at most $\delta$ that \emph{both} A) and B) occur simultaneously: 
\begin{enumerate}[label=\Alph*)]
    \item The verifier does not abort; AND
    \item Either $\mathcal{C}$ belongs to \emph{Q-CIRCUIT}$_{Y, q}$ but the verifier rejects it \emph{OR} $\mathcal{C}$ belongs to \emph{Q-CIRCUIT}$_{N, q}$ but the verifier accepts it. 
\end{enumerate}
\item  If the prover is honest, then the probability that the verifier aborts is at most $\delta$. 
\end{enumerate}
In particular, the probability of A) and B) occurring simultaneously is at most $2e^{-\Omega(N)}$.
\end{theorem}

\begin{proof} Let $p:=p^{\textnormal{noise}}$ and let $(S_C, S_{\textsf{X}}, S_{\textsf{Z}})$ be a random partition of $\{1, \cdots, N\}$ obtained by  rolling a fair three-sided die $N$ times. Following \Cref{AttackDeviationSec}, it suffices to consider an $(N, w, \mathcal{C})$-run of IPS3. If we obtain the desired probability bounds on each of the events in 1) and 2) for a $(N, w, \mathcal{C})$-run of IPS3, then we also obtain the same upper bounds on these events for a $(N, w, \mathcal{C})$-run of IPS1, since:
\begin{itemize}
\item A single run of IPS3 has the same completeness parameter as IPS1; AND
\item An upper bound on the soundness parameter for a single run of IPS3 is an upper bound on the soundness parameter for a single run of IPS2 and IPS1. 
\end{itemize}

We now definite intermediate variables $A, A'$ of, corresponding to their usage in \Cref{IfWSmallLowerBound}, as follows:
$$A:=\frac{\alpha^2N}{6\log(2/\delta)}$$
$$A':=\frac{\left(2+\frac{3}{A}\right)-\sqrt{\left(2+\frac{3}{A}\right)^2-4}}{2}$$
We then set $f(N, \alpha, \delta):=A'(1-A^{-1/2})$ and $g(N, \alpha, \delta):=A^{-1/2}/2$.

\bigskip

Note that $A'=1-O\left(1/A\right)$, so we indeed have $f(N, \alpha, \delta)=1-O\left(\frac{\sqrt{\log(2/\delta)}}{\alpha\sqrt{N}}\right)$.

\bigskip

We first show that, if the prover is honest, then the probability that the verifier aborts is at most $\delta$.

\begin{Claim}\label{FirstIntermediateCompleteness} If the prover is honest, then, for either type of the test rounds, the probability of failing at least a $w$-fraction of these rounds is at most $\delta/2$. \end{Claim}

\begin{claimproof} Choose one of the two types of these rounds and let $S\in\{S_{\textsf{X}}, S_{\textsf{Z}}\}$ be the corresponding set of rounds. Let $\mathcal{X}$ be a random variable taking values in $\{0, \cdots, |S|\}$, where $\textnormal{Pr}(\mathcal{X}=k)$ is the probability of passing $k$ of the $|S|$ test rounds. If the prover is honest, then any given round is passed with probability $1-p$. Thus, $\mathcal{X}\sim\textnormal{Binomial}(|S|, 1-p)$. We need to show that $\textnormal{Pr}(\mathcal{X}\leq |S|(1-w))\leq\frac{\delta}{2}$.

\bigskip

\textbf{Remark}: To apply the Hoeffding Bound, we want $|S|(1-w)\leq |S|(1-p)$. Actually, we need the gap between $(1-w)$ and $1-p$ to be bounded from below. This is why we have the lower bound on $w$ in the statement of \Cref{MainThmMultiRunStatement}. The upper bound on $w$ guarantees that the probability of simultaneously passing many test rounds and corrupting many computation rounds is low. The lower bound on $w$ guarantees that the probability that at least a $w$-fraction of test rounds is failed is low. Now, let $\beta:=(1-p)-(1-w)=w-p$. Our bound on $p$ implies that $|S|(1-w)\leq |S|(1-p)$, so, applying \cref{HoeffdingBd1}, we get
\begin{inequality}\label{HBoundCompl12}\textnormal{Pr}[\mathcal{X}\leq |S|(1-w)]\leq\textnormal{exp}\left(-2|S|\cdot [w-p]^2\right)\leq\textnormal{exp}\left(-2|S|\beta^2\right)\end{inequality}

Now, we need to show that $|S|$ is high with high probability. For each $t\in [0,1]$, let $E_t$ denote the event that $|S|>Nt$ and let $q_t$ denote the probability that $E_s$ occurs. Since each $q_t$ lies in $[0,1]$, it follows from \cref{HBoundCompl12} that

$$\textnormal{Pr}[\mathcal{X}\leq |S|(1-w)]\leq\min_{t\in [0,1]} \left(e^{-Nt\beta^2}+(1-q_t)\right)$$

We now bound $1-q_t$ from above for $t$ within a certain interval. We regard the generation of $S$ as a sequence of $N$ independent Bernoulli trials, each with a probability of success of $\frac{1}{3}$. It follows from \cref{HoeffdingBd1} that, if $t\in [0, 1/3]$, then there is a probability of at most $e^{-2((N/3)-tN)^2/N}$ that $E_t$ does not occur. Thus, we obtain 
$$\textnormal{Pr}[\mathcal{X}\leq |S|(1-w)]\leq\min_{t\in [0,1/3]} \left(e^{-2Nt\beta^2}+e^{-2N(t-1/3)^2}\right)$$
Taking the midpoint $t=1/6$, we get
\begin{inequality}\label{BoundGap1-w}\textnormal{Pr}[\mathcal{X}\leq |S|(1-w)]\leq e^{-N\beta^2/3}+e^{-N/18}\end{inequality}
The dominating term in \cref{BoundGap1-w} is $e^{-N\beta^2/3}$. Now, our bound on $p$ and choice of $w$ imply that $\beta^2\geq\frac{1}{4A}$. Furthermore, since $\alpha\in [0, 1/2)$, we get $\frac{N\beta^2}{3}\geq 2\log(2/\delta)$. We have $N/18\geq 2\log(2/\delta)$ as well, so we get
$$e^{-N\beta^2/3}+e^{-N/18}\leq 2\left(\frac{\delta}{2}\right)^2<\frac{\delta}{2}$$
so we are done. \end{claimproof}

\bigskip

We now finish the proof of \Cref{MainThmMultiRunStatement}. It follows from \Cref{FirstIntermediateCompleteness} that, if the prover is honest, then there is a probability of at most $\delta$ that the verifier aborts. Thus, 2) of \Cref{MainThmMultiRunStatement} holds. Now we prove 1). Suppose that both A) and B) occur. Thus, the verifier does not abort, and precisely one of the following holds:
\begin{enumerate}[label=\roman*)]
    \item Our circuit is a YES-instance of Q-CIRCUIT$_q$ but the verifier obtains the classical outcome $0$ in strictly fewer than half of the computation rounds; OR 
    \item Our circuit is a NO-instance of Q-CIRCUIT$_q$ but the verifier obtains the classical outcome $0$ in at least half of the computation rounds. 
\end{enumerate}

A single computation run without any honest noise or attacks from the prover will produce an output bit with the same distribution as measuring the last register of $\mathcal{C}|0\rangle^{\otimes n}$. If $N'$ is the number of computation rounds with no honest noise or attacks from the prover, then, since one of i)-ii) holds, as in \cref{EquivFrac}, we have
$$|S_C|-N'\leq\frac{|S_C|(2q-1)}{2q-2}$$
On the other hand, since the verifier does not abort, it necessarily holds that at most a $w$-fraction of the $(N-|S_C|)$ test rounds have errors which are detected by the verifier. By  \cref{Epsilon0Prot}, we can apply \Cref{IfWSmallLowerBound} with $\epsilon=0$, and it follows that there is a probability of at most $\delta$ that both of A)-B) of 1) occur, so we are done. This proves \Cref{MainThmMultiRunStatement}. \end{proof} 

\section{Concluding Remarks and Open Questions}

We have shown that, in the specified error model, an $N$-run protocol that interleaves test rounds and computation rounds can achieve a noise-tolerance of $\left(1-O\left(\frac{1}{\sqrt{N}}\right)\right)\frac{2q-1}{2q-2}$ for a polynomial-time quantum server with inherent error probability $q$. In this work, we have considered a very simple error model for the interactive proof setting, so a natural follow-up line of investigation consists of extending this to a more complex error model. In particular, in the circuit model, a delegated computation between an almost-classical verifier and a quantum server is subject to noise that arises from  the fidelity of the prover's gates, the auxiliary qubits prepared and sent by the client, and the quantum channel for sending these qubits. A client may wish to know how much noise in each of these individual components is tolerable, and how many rounds of a repeated interactive protocol they need in order to obtain a desired upper-bound on the probability of the unwanted outcome occurring in 1) of \Cref{MainThmMultiRunStatement}. Some recent work (\cite{ArxPrLeaksFaultTol}) has addressed this question for errors on the verifier side in the MBQC model, where noise may leak the classical secret control bits used by the verifier.

\section{Acknowledgements} 
We thank Peter Drmota and Dominik Leichtle for very interesting discussions in the early and late stages of this work respectively. We acknowledge funding support under the Digital Horizon Europe project FoQaCiA, Foundations of quantum computational advantage, grant no.~101070558, together with the support of the Natural Sciences and Engineering Research Council of Canada (NSERC)(ALLRP-578455-2022).

\bibliographystyle{alphaarxiv.bst}
\bibliography{quasar-full.bib,  quasar-abrv.bib, 
quasar.bib, quasar-more-noisy-verification.bib}

\begin{appendix}
\counterwithout{equation}{section}
\addtocounter{equation}{0}
\section{The proof of the main probability lemma}

\begin{lemma}\label{IfWSmallLowerBound} Let $C$ be a random subset of $\{1, \cdots, N\}$ which is generated by a sequence of $N$ independent Bernoulli trials, each with a $\frac{1}{3}$-probability of success. Fix real parameters $\epsilon\in [0,1]$ and $\alpha, \delta\in [0, 1/2)$ and $A\geq 100$. Fix a subset $S$ of $\{1, \cdots, N\}$ and suppose that, for each $x\in S$, we flip an $\epsilon$-coin. For each $x\in S$, let $\mathcal{Y}_x$ be a random variable corresponding to the coin flip, that is $\textnormal{Pr}(\mathcal{Y}_x=0)=\epsilon$ and $\textnormal{Pr}(\mathcal{Y}_x=1)=1-\epsilon$. Then, as long as the following two bounds are satisfied:
\begin{inequality}\label{BoundonNEpsAlphDelt} N\geq\frac{6A\log(2/\delta)}{(1-\epsilon)^2\alpha^2}\end{inequality}
\begin{inequality}\label{BoundOnWEpsALphDelt} w\leq \left(1-A^{-1/2}\right)A'(1-\epsilon)\alpha \end{inequality}
where $$A':=\frac{\left(2+\frac{3}{A}\right)-\sqrt{\left(2+\frac{3}{A}\right)^2-4}}{2}$$
the probability that both of the following hold simultaneously is bounded from above by $\delta$.
\begin{enumerate}[label=\alph*)]
\item $|S\cap C|>\alpha|C|$
\item $\{x\in S\setminus C:\mathcal{Y}_x=1\}|\leq w(N-|C|)$
\end{enumerate}
\end{lemma}

\begin{proof} For each $s\in [0,1]$, let $E_s$ denote the event $1-\frac{|C|}{N}>s$ and $p_s$ denote the probability that $E_s$ occurs. 

\bigskip

\begin{Claim}\label{If1-sBoundthen1-psbound} 
Both of the following hold:
\begin{enumerate}[label=\roman*)]
\item For $1-s\geq 1/3$, we have $1-p_s \leq e^{-2N(s-2/3)^2}$.
\item For $1-s\leq 1/3$, we have $p_s\leq e^{-2N(s-2/3)^2}$ 
\end{enumerate}\end{Claim}

\begin{claimproof} We regard the generation of $C$ as a sequence of $N$ independent coin flips, where the coin takes value~1 with probability $1/3$ and 0 with probability $2/3$. If $1-s\geq 1/3$, then, by  \cref{HoeffdingBd2}, the probability that the number of successes of this sequence of flips is at at least $N(1-s)$ is bounded from above by $\textnormal{exp}\left(-2N(s-2/3)^2\right)$. As $E_s$ is the event that there are strictly fewer than $N(1-s)$ successes, it follows that $1-p_s \leq e^{-2N(s-2/3)^2}$. This proves i). An identical argument, applying \cref{HoeffdingBd1}, proves ii). \end{claimproof}

\bigskip

For each $\ell\in [0,1]$, let $F_{\ell}$ denote the event that $|S\setminus C|\geq\ell (N-|C|)$ and $\overline{F}_{\ell}$ denote the event that $|S\setminus C|<\ell (N-|C|)$. We show that there exists an $\ell^*\in [0,1]$ such that if $F_{\ell^*}$ holds, then a) fails with high probability and, if $\overline{F}_{\ell^*}$ holds, then b) fails with high probability so, in any case, the probability of both of them holding is low.

\bigskip  

\begin{Claim}\label{IfSlargeThenFellCLow} For $\ell, m\in \left[0, 1\right]$ with $\ell<m$, We have $$\textnormal{Pr}\left[|S\setminus C|<\ell (N-|C|)\bigg\vert |S|\geq mN\right]\leq e^{-2N/9}+e^{-2mN\left(1-\frac{\ell}{m}\right)^2/9} $$ 
 \end{Claim}

\begin{claimproof}  For the purpose of \Cref{IfSlargeThenFellCLow}, we temporarily regard the generation of $S$ as randomly drawing $|S|$ elements from $[N]$, where $|S|\geq mN$ and the elements of $[N]\setminus C$ are marked. Let $\mathcal{Z}$ be a random variable corresponding to $|S\setminus C|$. We apply \Cref{TailBdCor1}. Let $\lambda=\ell(N-|C|)$. Since $|S|\geq mN$ We then have
$$\ell(N-|C|)<\frac{|S|(N-|C|)}{N}$$
Now, if $\lambda=0$, then either $\ell=0$ or $N=[C]$, and thus $\textnormal{Pr}[\overline{F}_{\ell}]=0$, since $|S|\geq mN>0$. Thus, by \Cref{TailBdCor1}, we get
$$\textnormal{Pr}\left[\overline{F}_{\ell}\bigg\vert |S|\geq mN\right]\leq \textnormal{exp}\left(-2|S|\left(\frac{N-|C|}{N}-\frac{\ell (N-|C|)}{|S|}\right)^2\right)$$
Furthermore, we have $\frac{N-|C|}{N}-\frac{\ell (N-|C|)}{|S|}>0$, so we get
$$|S|\left(\frac{N-|C|}{N}-\frac{\ell (N-|C|)}{|S|}\right)^2\geq\frac{|S|(N-|C|)^2}{N^2}\left(1-\frac{\ell}{m}\right)^2\geq\frac{m(N-|C|)^2}{N}\left(1-\frac{\ell}{m}\right)^2$$
Thus,
$$\textnormal{Pr}\left[\overline{F}_{\ell}\bigg\vert |S|\geq mN\right]\leq\textnormal{exp}\left(-2m\frac{(N-|C|)^2\left(1-\frac{\ell}{m}\right)^2}{N}\right)$$
Thus, conditioning on $E_s$, we have the following for each $s\in [0,1]$:
$$\textnormal{Pr}\left[\overline{F}_{\ell}\bigg\vert |S|\geq mN\ \textnormal{and}\ E_s\right]<e^{-2s^2mN\left(1-\frac{\ell}{m}\right)^2}$$
Since this holds for each $s\in [0,1]$ and each $p_s$ lies in $[0,1]$, we have the following by i) of \Cref{If1-sBoundthen1-psbound}
$$\textnormal{Pr}\left[\overline{F}_{\ell}\bigg\vert |S|\geq mN\right]<\min_{0\leq s\leq 1}\left((1-p_s)+p_s\cdot e^{-2s^2mN\left(1-\frac{\ell}{m}\right)^2}\right)$$
$$\leq \min_{0\leq s\leq 2/3}\left(e^{-2N(s-2/3)^2}+e^{-2s^2mN\left(1-\frac{\ell}{m}\right)^2}\right)$$
Taking the midpoint $s=1/3$, we obtain the desired bound. This proves \Cref{IfSlargeThenFellCLow}. \end{claimproof}

\bigskip

Analogously, we have the following:

\begin{Claim}\label{IntermedCLBound1} For $\ell, m\in \left[0, 1\right]$ with $\ell\geq m>0$, We have $$\textnormal{Pr}\left[|S\cap C|>\ell |C| \, \bigg\vert \, |S|\leq mN\right]\leq e^{-N/18}+e^{-|S|\left(\frac{\ell}{m}-1\right)^2/18}$$ 
 \end{Claim}

\begin{claimproof}  For the purpose of \Cref{IfSlargeThenFellCLow}, we temporarily regard the generation of $S$ as randomly drawing $|S|$ elements from $[N]$, where $|S|\leq mN$ and the elements of $C$ are marked. Let $\mathcal{Z}'$ be a random variable corresponding to $|S\cap C|$. We apply \Cref{TailBdCor2}. Let $\lambda=\ell |C|$. Note that $\ell |C|>\frac{|S|\cdot |C|}{N}$, so we obtain
$$\textnormal{Pr}\left[\mathcal{Z}'>\ell |C|\,\bigg\vert \,|S|\leq mN\right]\leq\textnormal{exp}\left(-2|S|\left(\frac{\ell |C|}{|S|}-\frac{|C|}{N}\right)^2\right)$$
Note that we have $\frac{\ell|C|}{|S|}-\frac{|C|}{N}>0$, so we obtain
$$\frac{\ell|C|}{|S|}-\frac{|C|}{N}>\frac{|C|}{N}\left(\frac{\ell}{m}-1\right)\geq 0$$
and thus
$$|S|\left(\frac{|C|}{N}-\frac{\ell|C|}{|S|}\right)^2\geq |S|\cdot\frac{|C|^2(\frac{\ell}{m}-1)^2}{N^2}\geq 0$$
and we have
$$\textnormal{Pr}\left[\mathcal{Z}'>\ell |C|\, \bigg\vert\, |S|\leq mN\right]\leq\textnormal{exp}\left(-2|S|\cdot\frac{|C|^2(\frac{\ell}{m}-1)^2}{N^2}\right)$$
Thus, conditioning on $\overline{E}_s$, we have the following for each $s\in [0,1]$:
$$\textnormal{Pr}\left[\mathcal{Z}'>\ell |C|\, \bigg\vert \,|S|\leq mN\ \textnormal{and}\ \overline{E}_s\right]\leq\textnormal{exp}\left(-2|S|\cdot (1-s)^2\left(\frac{\ell}{m}-1\right)^2\right)$$
Since this holds for each $s\in [0,1]$ and each $1-p_s$ lies in $[0,1]$, we have the following by ii) of \Cref{If1-sBoundthen1-psbound}
$$\textnormal{Pr}\left[\mathcal{Z}'>\ell|C|\,\bigg\vert \,|S|\leq mN\ \textnormal{and}\ \overline{E}_s\right]\leq\min_{0\leq s\leq 1}\left(p_s+(1-p_s)\cdot e^{-2|S|(1-s)^2N\left(\frac{\ell}{m}-1\right)^2}\right)$$

$$\leq \min_{\frac{2}{3}\leq s\leq 1}\left(e^{-2N\left(s-2/3\right)^2}+e^{-2|S|(1-s)^2\left(\frac{\ell}{m}-1\right)^2}\right)$$
Taking the midpoint $s=5/6$, we obtain the desired bound. 
\end{claimproof}

\bigskip

We now consider event b). We have a sequence of $|S\setminus C|$ Bernouilli trials, where a success for trial $x$ is defined as obtaining $\mathcal{Y}_x=1$. Let $\mathcal{Y}$ be a random variable denoting the number of successes among these $|S\setminus C|$ trials. We show that, if $|S\setminus C|$ is a large fraction of $N-|C|$, then b) fails with high probability.

\bigskip

\begin{Claim}\label{BoundonB)ForFEll} For any $\ell, m\in \left[0, 1\right]$ with $m\leq\ell(1-\epsilon)$ , we have $$\textnormal{Pr}\left[\mathcal{Y}\leq m(N-|C|)\,\bigg\vert\, F_{\ell}\right]\leq e^{-2N/9}+e^{-2N(\ell(1-\epsilon)-m)^2/3}$$
\end{Claim}

\begin{claimproof} By assumption, we have $|S\setminus C|\geq \ell\cdot (N-|C|)$.  Since $m(N-|C|)\leq \ell(N-|C|)(1-\epsilon)$, we can apply \cref{HoeffdingBd1} to obtain

$$\textnormal{Pr}\left[\mathcal{Y}\leq m(N-|C|)\,\bigg\vert \,F_{\ell}\right]\leq \textnormal{exp}\left(-2\frac{(\ell(N-|C|)(1-\epsilon)-(N-|C|)m)^2}{N-|C|}\right)$$

Thus,

$$\textnormal{Pr}\left[\mathcal{Y}\leq m(N-|C|)\,\bigg\vert \,F_{\ell}\right]\leq \textnormal{exp}\left(-2(N-|C|)(\ell(1-\epsilon)-m)^2\right)$$

 For any $s\in [0,1]$, if $E_s$ occurs, then we obtain
$$\textnormal{Pr}\left[\mathcal{Y}\leq m(N-|C|)\,\bigg\vert \,F_{\ell}\ \textnormal{and}\ E_s\right]\leq\textnormal{exp}\left(-2sN\left(\ell(1-\epsilon)-m\right)^2\right)$$

Thus,

$$\textnormal{Pr}\left[\mathcal{Y}\leq m(N-|C|)\,\bigg\vert \,F_{\ell}\right]\leq \min_{s\in [0,1]}\left((1-p_s)+p_s\cdot e^{-2sN(\ell(1-\epsilon)-m)^2}\right)$$
$$\leq \min_{0\leq s\leq 2/3}\left(e^{-2N(s-2/3)^2}+e^{-2sN(\ell(1-\epsilon)-m)^2}\right)$$
where the rightmost inequality is obtained from  i) of \Cref{If1-sBoundthen1-psbound}. Taking the midpoint $s=1/3$, we obtain the desired bound. This proves \Cref{BoundonB)ForFEll}. \end{claimproof}

\bigskip

Applying \Cref{IntermedCLBound1}, we now show that, if $|S|$ is small, then a) holds with low probability. To prove this we first prove the following intermediate result.

\begin{Claim}\label{IntermResUpperLower} Let $x\in (0,1]$. Then $$\textnormal{Pr}\left[|S\cap C|>\alpha |C|\,\bigg\vert\, |S|/N\leq \alpha x\right]\leq e^{-N/18}+e^{-\left(\frac{1}{x}-1\right)^2\alpha xN/18}$$ \end{Claim}

\begin{claimproof}  We apply \Cref{IntermedCLBound1} with $\ell=\alpha$ and $m=\alpha x$, so we get
$$\textnormal{Pr}\left[\textnormal{a) holds}\,\bigg\vert\, |S|\leq mN\right]\leq e^{-N/18}+e^{-\left(\frac{\ell}{m}-1\right)^2|S|/18}$$
Thus, we have the following bound in the case where we have the exact value of $|S|/N$.
\begin{inequality}\label{IntermedStepEqual1}\textnormal{Pr}\left[\textnormal{a) holds}\,\bigg\vert\, |S|/N=\alpha x\right]\leq e^{-N/18}+e^{-\left(\frac{1}{x}-1\right)^2\alpha xN/18}\end{inequality}

Note that we have the following:
$$\textnormal{Pr}\left[\textnormal{a) holds}\bigg\vert |S|/N\leq\alpha x\right]\leq\max\left\{\textnormal{Pr}\left[\textnormal{a) holds}\bigg\vert |S|=\alpha Ny\right]: y\in\left\{0, \frac{1}{\alpha N}, \frac{2}{\alpha N}, \cdots, \frac{\left\lfloor\alpha N x\right\rfloor}{\alpha N}\right\} \right\}$$
If $|S|=0$, then a) holds with probability zero. Furthermore, since the function $y\rightarrow \left(\frac{1}{y}-1\right)^2y$ is monotone decreasing over $(0,1]$ and $\frac{\lfloor\alpha Nx\rfloor}{\alpha N}\leq x$, it follows from \cref{IntermedStepEqual1} that 
$$\textnormal{Pr}\left[\textnormal{a) holds}\bigg\vert |S|/N\leq\alpha x\right]\leq e^{-N/18}+e^{-\left(\frac{1}{x}-1\right)^2\alpha xN/18}$$
as desired.
\end{claimproof}

Applying \Cref{IntermResUpperLower}, we have the following:

\begin{Claim} If $|S|/N\leq A'\alpha$ then a) holds with probability at most $\delta$.
\end{Claim}

\begin{claimproof} We first note that $e^{-N/18}\leq\delta/2$, since $N\geq 18\log(2/\delta)$, so now it suffices to show that $x=A'$ satisfies the inequality
$$e^{-\left(\frac{1}{x}-1\right)^2\alpha xN/18}\leq\delta/2$$

We want $x\in (0,1]$ to be as large as possible subject to the above constraint. Equivalently, we want $x\in (0, 1]$ to be as large as possible subject to
$$\alpha x\left(\frac{1}{x}-1\right)^2 N/18\geq \log(2/\delta)$$
Since $N\geq\frac{6A\log(2/\delta)}{\alpha^2}$, and since $\alpha\in [0,1/2)$, a sufficient condition for this inequality to hold is that
\begin{inequality}\label{QuadratAtEq}x\left(\frac{1}{x}-1\right)^2\geq\frac{3}{A}\end{inequality}
Since $x\left(\frac{1}{x}-1\right)^2$ is monotone-decreasing over $(0,1]$ we just need to solve \cref{QuadratAtEq} at equality, so we have
$$x=\frac{(2+3/A)\pm\sqrt{\left(2+\frac{3}{A}\right)^2-4}}{2}$$
Since $0<x<1$, we have $x=A'$, as desired. \end{claimproof}

\bigskip

Thus, if $|S|/N\leq A'\alpha$ then we are done. Now suppose that $|S|/N>A'\alpha$. We show that b) holds with probability at most $\delta$, and then we are done. For each $\ell\in [0, 1]$, let $q_{\ell}$ denote the probability that $F_{\ell}$ occurs. Since $|S|/N>A'\alpha$, we can apply \Cref{IfSlargeThenFellCLow} to any $\ell\in\left[0, A'\alpha\right)$. Likewise, we can apply \Cref{BoundonB)ForFEll} to any $\ell\in\left[w/(1-\epsilon), 1\right]$. Now, our bound on $w$ implies that $\frac{w}{1-\epsilon}<A'\alpha$, so $\left(\frac{w}{1-\epsilon}, A'\alpha\right)\neq\varnothing$, and, for each $\ell\in (w/(1-\epsilon), A'\alpha)$, we have
$$\textnormal{Pr}\left[\mathcal{Y}\leq w(N-|C|)\right]\leq\textnormal{Pr}\left[ \overline{F}_{\ell}\right]+\textnormal{Pr}\left[\mathcal{Y}\leq w(N-|C|)\ \textnormal{and}\ F_{\ell}\right]$$
$$=(1-q_{\ell})+q_{\ell}\cdot\textnormal{Pr}\left[\mathcal{Y}\leq w(N-|C|) \bigg\vert F_{\ell}\right]$$
\begin{inequality}\label{BoundToLeft}\leq\left(e^{-2N/9}+e^{-2N(A'\alpha-\ell)^2/(9A'\alpha)}\right)+\left(e^{-2N/9}+e^{-2N\left(\ell(1-\epsilon)-w\right)^2/3}\right)\end{inequality}
We need the interval $\left(\frac{w}{1-\epsilon}, A'\alpha\right)$ to be sufficiently long. Let $L$ be the radius of this interval. The bound above in \cref{BoundToLeft}, with $\ell$ being the midpoint of the interval $\left(\frac{w}{1-\epsilon}, A'\alpha\right)$, implies that 
$$\textnormal{Pr}\left[\mathcal{Y}\leq w(N-|C|)\right]\leq 2e^{-2N/9}+e^{-2NL^2/(9A'\alpha)}+e^{-2NL^2(1-\epsilon)^2/3}$$
Note that $L=\frac{1}{2}\left(A'\alpha-\frac{w}{1-\epsilon}\right)$, so $\frac{L}{A'\alpha}=\frac{1}{2}\left(1-\frac{w}{(1-\epsilon)A'\alpha}\right)$, so our bound in \cref{BoundOnWEpsALphDelt} implies that $L/A'\alpha\geq\frac{1}{\sqrt{A}}$. Thus, we get
\begin{inequality}\label{ThreeTermEq}\textnormal{Pr}\left[\mathcal{Y}\leq w(N-|C|)\right]\leq 2e^{-2N/9}+e^{-(2NA'\alpha )/(9A)}+e^{-2N(A'\alpha)^2(1-\epsilon)^2/(3A)}\end{inequality}
Bounding everything in \cref{ThreeTermEq} is just a matter of $N$ being large enough. In particular, note that we can bound $N\alpha^2(1-\epsilon)^2/A$. The dominating term in \cref{ThreeTermEq} is the rightmost term. 

\begin{Claim}\label{BoundonDomTerm} $e^{-N(A'\alpha)^2(1-\epsilon)^2/(3A)}\leq\delta/2$. \end{Claim}

\begin{claimproof} Note that $A$ is sufficiently large that $A'\geq 4/5$ It suffices that
$$N(A')^2\geq\frac{3A\log(2/\delta)}{(1-\epsilon)^2}$$
Note that $A$ is sufficiently large that $A'\geq\frac{1}{\sqrt{2}}$, so it follows from \cref{BoundonNEpsAlphDelt} that this bound indeed holds.
\end{claimproof}

\bigskip

Now we just bound the sum of the other two smaller terms:

\bigskip

\begin{Claim}\label{SumTwoBoundDeltaHalf} $2e^{-2N/9}+e^{-(2NA'\alpha)/(9A)}\leq\delta/2$. \end{Claim}

\begin{claimproof} The dominating term is the term on the right. Note that $A$ is sufficiently large that $2A'/9\geq 1/6$. Thus, we get
$$e^{-(-2NA'\alpha)/(9A)}\leq e^{-N\alpha/(6A)}$$
Now, our bound on $N$ implies that $\frac{N\alpha}{6A}\geq\log(2/\delta)/\alpha$. Since $\alpha\in [0,1/2)$, we get $e^{(-(2NA'\alpha)/(9A)}\leq\left(\frac{\delta}{2}\right)^2$. Now we consider the other term. Since $\alpha^2\leq\frac{1}{4}$, we get $N\geq 2400\log(2/\delta)$, so $2e^{-2N/9}\leq 2\left(\frac{\delta}{2}\right)^{500}$. Thus, we get
$$2e^{-2N/9}+e^{-(2NA'\alpha)/(9A)}\leq+2\left(\frac{\delta}{2}\right)^{500}+\left(\frac{\delta}{2}\right)^2\leq\frac{\delta}{2}$$ as desired.
\end{claimproof}

\bigskip

Applying \Cref{BoundonDomTerm} and \Cref{SumTwoBoundDeltaHalf} to \cref{ThreeTermEq}, we are done. \end{proof}
\end{appendix}

\end{document}